\documentclass[12pt,a4paper]{article}

\usepackage{geometry}
\usepackage[pagebackref=false]{hyperref}
  \hypersetup{
    colorlinks=true,
    linkcolor=blue,
    citecolor=green,
    filecolor=magenta,
    urlcolor= blue
    }

\usepackage{amsfonts}
\usepackage{amsmath}
\usepackage{amssymb}

\usepackage[amsmath,thmmarks]{ntheorem}
\theoremheaderfont{\bfseries}
\theorembodyfont{\itshape}
\theoremstyle{plain}
\theoremseparator{:}
\theoremsymbol{}

\newtheorem{thm}{Theorem}
\newtheorem{cor}{Corollary}
\newtheorem{alg}{Algorithm}
\newtheorem{prot}{Protocol}
\theorembodyfont{\upshape}
\newtheorem{defn}{Definition}

\theoremsymbol{\ensuremath{\blacksquare}}
\newtheorem{expl}{Example}
\theoremheaderfont{\itshape}
\theorembodyfont{\upshape}
\theoremstyle{nonumberplain}
\theoremsymbol{\ensuremath{\square}}
\newtheorem{proof}{Proof}

\begin{document}

\title{Public Key Protocols over the Ring $E_{p}^{(m)}$%
  }

\author{Joan-Josep Climent%
\thanks{Departament de Matem\`{a}tiques, Universitat d'Alacant, email: \texttt{jcliment@ua.es}},
\\
Juan Antonio L\'{o}pez-Ramos%
\thanks{Departamento de Matem\'{a}ticas, Universidad de Almer\'{\i}a, email: \texttt{jlopez@ual.es}},
}


\maketitle

\begin{abstract}
In this paper we use the nonrepresentable ring $E_{p}^{(m)}$ to introduce public key cryptosystems in noncommutative settings and based on the Semigroup Action Problem and the Decomposition Problem respectively.
\end{abstract}

\section{Introduction}
\label{Intro}

Most public-key cryptosystems are based on certain specific problems of number theory. 
One of these problems is the \emph{Integer Factorization Problem} (IFP) over the ring $\mathbb{Z}_{n}$, being $n$ the product of two large prime numbers; the well known cryptosystem RSA \cite{Rivest1978} is based in this problem. 
The second classical problem is the \emph{Discrete Logarithm Problem} (DLP) over a finite field $\mathbb{Z}_{p}$, being $p$ a large prime; the Diffie-Hellman key exchange protocol \cite{Diffie1976} and ElGamal protocol \cite{ElGamal1985} are based on this problem. 
In general, we can say that their robustness depends on the computational difficulty of solving certain mathematical problems over finite commutative algebraic structures. 
Some efficient attacks based on the commutative property of these structures are well known: Quadratic Sieve, General Number Field Sieve, Pollard's rho algorithm, Index-Calculus, etc.\ (see, for example, \cite{Menezes1996bk,Stinson1995bk}, and the references within these books).

This fact, together with the increase of the computational power of modern computers, has made these techniques become more and more insecure.
As a result there exists an active field of research known as noncommutative algebraic cryptography (see, for example, \cite{Anshel1999,Ko2007,Ko2000,Sakalauskas2003,Sidelnikov1994}) aiming to develop and analyse new cryptosystems and key exchange protocols based on noncommutative cryptographic platforms.
A very good exposition of the problems underlying in this noncommutative approach can be found in \cite{Myasnikov2008bk} and some of them are the following, where $G$ is a nonabelian group:

\begin{itemize}
\item \emph{Conjugator Search Problem (CSP):}
Given $(x,y) \in G \times G$, the problem is to find $z \in G$ such that $y = z^{-1} x z$.
\item \emph{Decomposition Problem (DP):}
Given $(x,y) \in G \times G$ and $S \subseteq G$, the problem is to find $z_{1}, z_{2} \in S$ such that $y = z_{1} x z_{2}$.
\item \emph{Symmetrical Decomposition Problem (SDP):}
Given $(x,y) \in G \times G$ and $m,n \in \mathbb{Z}$, the problem is to find $z \in G$ such that $y = z^{m} x z^{n}$.
\item \emph{Generalized Symmetrical Decomposition Problem (GSDP):}
Given $(x,y) \in G \times G$, $S \subseteq G$ and $m, n \in \mathbb{Z}$, the problem is to find $z \in S$ such that $y = z^{m} x z^{n}$.
\end{itemize}

Note that for the DP we can assume that $G$ is a semigroup, since in this case we do not need invertible elements.

Several authors have proposed and used certain nonabelian groups for
key exchange problems. In
\cite{Anshel2001,Anshel1999,Ko2007,Ko2000}, the authors suggest to
use braid groups as platform groups for their respective protocols.
In \cite{Paeng2001}, the authors propose a public key cryptosystem
whose security is based on the DLP for the automorphism defined by
the conjugation operation and the difficulty to find the conjugate
element on finite nonabelian groups. In \cite{Shpilrain2006}, the
authors suggest the use of a finite representation of a nonabelian
group, called Thomson's group, to develop a public key cryptosystem,
where they raised for the first time the difficulty to find a
solution for the SDP. Finally, in \cite{Thomas2008}, the authors
propose a cryptosystem whose robustness is based on the difficulty
to solve the CSP and SDP over any noncommutative algebraic
structure.

In the noncommutative setting, we can find different implementations based on the Diffie-Hellman protocol in matrix rings, for different kind of matrices \cite{Climent2006b,Stickel2005,Yoo2000}.
A detachable recent work in this setting is \cite{Mahalanobis2013} where the author shows the usability of a certain class of matrices that arise a noncommutative framework to define the ElGamal cryptosystem and the corresponding Diffie-Hellman protocol as well as its eficiency.
Another system based on the discrete logarithm problem in the authomorphisms group is the so-called MOR cryptosystem, recently described for authomorphisms of $p$-groups whose order is coprime to $p$ in \cite{Mahalanobis2008b,Mahalanobis2015}.

Moreover, the idea to develop systems of open distribution keys as well as session key exchange protocols, on the basis of noncommutative (semi)groups, is present in \cite{Mahalanobis2008,Sakalauskas2003,Sidelnikov1994}.

Finally, with the idea to generalize the protocols based on groups
and take advantage of the difficulty of to solve the DLP in these
groups, Shpilrain and Zapata in \cite{Shpilrain2006b} give a general
framework using actions of groups to define key exchange protocols.
In this setting, Maze \emph{et al.} \cite{Maze2007} introduce the
\emph{Semigroup Action Problem} (SAP) as:
\begin{quote}
Let $G$ be a finite semigroup acting on a finite set $S$.
Given $x, y \in S$ with $y = g \cdot x$ for some $g \in G$, find $h \in G$ such that $y = h \cdot x$.
\end{quote}
In this way, any Abelian semigroup $G$ acting on a finite set $S$ provides a Diffie-Hellman key exchange protocol and an ElGamal protocol (cf. \cite{Maze2007}).

\medskip

Our aim in this paper is to provide public key cryptosystems in a noncommutative setting given by the ring $E_p^{(m)}$ based on the SAP and DP. 
The remainder of this paper is organized as follows. 
In Section~\ref{ringEpm} we remind some properties of the nonrepresentable ring $E_p^{(m)}$ that will be used through this paper. 
In Section~\ref{PKCSAP} we use an action of the ring $E_p^{(m)}$ to introduce a public key cryptosystem based on the SAP. 
Finally, in Section~\ref{PKCDP} we use this ring to give a public key cryptosystem based on the DP and we relate security of both introduced cryptosystems.

\section{The ring $E_{p}^{(m)}$}
\label{ringEpm}

Bergman \cite{Bergman1974} proved that $\mathrm{End}(\mathbb{Z}_{p} \times \mathbb{Z}_{p^{2}})$ is a semilocal ring that cannot be embedded in matrices over any commutative ring.
In \cite{Climent2011a} the authors showed that $\mathrm{End}(\mathbb{Z}_{p} \times \mathbb{Z}_{p^{2}})$
is isomorphic to a ring whose elements can be expressed as square matrices whose rows are given by elements in $\mathbb{Z}_{p}$ and $\mathbb{Z}_{p^{2}}$ respectively and with an arithmetic that is similar to matrix addition and multiplication.
Then they use this ring to define a key exchange protocol (see also \cite{Climent2011c,Climent2012c}).
However this was cryptoanalyzed in \cite{Kamal2012} using some invertible elements, which motivated the introduction of an extension of such a ring \cite{Climent2014d} such that almost all elements are noninvertible

This new ring is defined as the set (see \cite[Theorem~1]{Climent2014d})
\[
  E_{p}^{(m)}
  =
  \left\{
    [a_{ij}] \in \mathrm{Mat}_{m \times m} (\mathbb{Z})
    \ | \
    a_{ij} \in \mathbb{Z}_{p^{i}} \ \text{if} \ i \leq j,
    \ \text{and} \
    a_{ij} \in p^{i-j} \mathbb{Z}_{p^{j}} \ \text{if} \ i > j
  \right\}
\]
with the addition and the multiplication defined, respectively, as follows
\begin{gather*}
  \big[
    a_{ij}
  \big]
  +
  \big[
    b_{ij}
  \big]
  =
  \big[
    (a_{ij}+b_{ij}) \bmod{p^{i}}
  \big], \\
  \big[
    a_{ij}
  \big]
  \cdot
  \big[
    b_{ij}
  \big]
  =
  \left[
    \left(
      \sum_{k=1}^{m} a_{ik} b_{kj}
    \right)
    \bmod{p^{i}}
  \right].
\end{gather*}
Here $\mathrm{Mat}_{m \times m} (\mathbb{Z})$ denotes the set of $m \times m$ matrices with entries in $\mathbb{Z}$, and $p^{r} \mathbb{Z}_{p^{s}}$ denotes the set $\left\{p^{r} u \ | \ u \in \mathbb{Z}_{p^{s}} \right\}$ for positive integers $r$ y $s$.
Moreover, $\left| E_{p}^{(m)} \right| = p^{(2 m^{3} + 3 m^{2} + m)/6}$.

The following results on the ring $E_{p}^{(m)}$ can be found in \cite{Climent2014d}.

\begin{thm}[Theorem~3 of \cite{Climent2014d}]
Let $A = [a_{ij}] \in E_{p}^{(m)}$.
Then $A$ is invertible if and only if $a_{ii} \not\equiv_{p} 0$ for every $i=1,2,\dots,m$.
\end{thm}

As a consequence, the number of invertible elements in $E_{p}^{(m)}$ is $p^{(2 m^{3} + 3 m^{2} - 5 m)/6}$ (see \cite[Corollary~1]{Climent2014d}).
Moreover, the fraction of invertible elements in $E_{p}^{(m)}$ is
\begin{math}
  \left(
    \frac{p-1}{p}
  \right)^{m}.
\end{math}
So, an adequate election of $p$ and $m$ may result that the number of invertible elements in $E_{p}^{(m)}$ is almost zero (see \cite[Tables~1 and 2]{Climent2014d}).

Another needed property is the characterization of the center of the
ring $E_{p}^{(m)}$. The following result provides such a
characterization.

\begin{thm}[Page~359 of \cite{Climent2014d}] \label{centerEpm}
The center of $E_{p}^{(m)}$ is given by the set
\[
  Z(E_{p}^{(m)})
  =
  \left\{
    [a_{ij}] \in E_{p}^{(m)}
    \ | \
    a_{ii} = \sum_{j=0}^{i-1} p^{j} u_{j}, \ \text{with} \ u_{j} \in \mathbb{Z}_{p}
    \ \text{and} \
    a_{ij} = 0 \ \text{if} \ i \neq j
  \right\}.
\]
As a consequence the number of central elements in $E_{p}^{(m)}$ is $p^{m}$.
\end{thm}

\section{A public key cryptosystem based on the SAP}
\label{PKCSAP}

Our aim in this section is to use the arithmetic of the ring
$E_{p}^{(m)}$ and an action of this ring to define a public key
cryptosystem in a noncommutative setting.

Thus, based on the multiplication defined on the ring $E_{p}^{(m)}$ and the structure of its elements we may define an action of the ring $E_{p}^{(m)}$ over the set $\mathbb{Z}_{p} \times \mathbb{Z}_{p^{2}} \times \cdots \times \mathbb{Z}_{p^{m}}$.
As we will show, this action will arise a public key cryptosystem based on the SAP (cf.\ \cite{Maze2007}).

Let us consider the center $Z(E_{p}^{(m)})$ of the ring $E_{p}^{(m)}$.
For a given $M \in E_{p}^{(m)}$, let $\mathrm{Cen}(M)$ be the set of elements $X$ in $E_{p}^{(m)}$ such that $X M = M X$.
The algorithm is given by the following

\begin{alg} \label{sap}
Let $M \in E_{p}^{(m)}$ be a public value and $S \in \mathbb{Z}_{p} \times \mathbb{Z}_{p^{2}} \times \cdots \times \mathbb{Z}_{p^{m}}$ the message that Bob wants to send Alice.
Then:
\begin{enumerate}
\item \label{sap.1}
Alice chooses $R \in \mathbb{Z}_{p} \times \mathbb{Z}_{p^{2}} \times \cdots \times \mathbb{Z}_{p^{m}}$, $F \in \mathrm{Cen}(M)$ and computes $T = F \cdot R$.

\item \label{sap.2}
Alice makes public the pair $(R,T)$, keeping secret her private key $F$.

\item \label{sap.3}
Bob chooses randomly $G = \displaystyle \sum_{i=0}^{k} C_{i} M^{i}$, where $C_{i} \in Z(E_{p}^{(m)})$ and sends Alice the pair
\begin{math}
  (H, D)
  =
  (G \cdot R, S + G \cdot T).
\end{math}

\item \label{sap.4}
Alice gets the secret by computing $S = D - F \cdot H$.
\end{enumerate}
\end{alg}

Note that the commutativity of $F$ and $G$ gives the correctness of the preceding algorithm and that Alice and Bob could exchange the way they choose the elements appearing through it.

As it is asserted in the general case in \cite{Maze2007} we can observe that breaking the preceding algorithm involves solving the SAP, i.e., given the values $R$ and $T = F \cdot R \in \mathbb{Z}_{p} \times \mathbb{Z}_{p^{2}} \times \cdots \times \mathbb{Z}_{p^{m}}$, find $A \in E_{p}^{(m)}$ such that $T = A \cdot
R$.

To understand the size of the underlying SAP, let $p$ be a prime, consider $x \in \mathbb{Z}_{p^{m}}$ coprime to $p$, and let $n$ be the order of $x$.
Now let $M = [a_{ij}]$ be the element of $E_{p}^{(m)}$ given by
\[
  a_{ij}
  =
  \begin{cases}
    x, & \text{for $i=j=m$}, \\
    0, & \text{otherwise}.
  \end{cases}
\]
For $k=1,2,\ldots$, assume that $M^{k} = [a_{ij}^{(k)}]$; then
\[
  a_{ij}^{(k)}
  =
  \begin{cases}
    x^{k}, & \text{for $i=j=m$}, \\
    0,     & \text{otherwise}.
  \end{cases}
\]
Thus it is clear that $M^{n+1} = M$ and $M, M^{2}, \ldots, M^{n}$ are pairwise different.
Since $M^{0} = I$ (by definition), the elements of $G = \displaystyle \sum_{k=0}^{n} C_{k} M^{k}$ are in
$\mathrm{Cen}(M)$ and have the form
\begin{equation} \label{eqG}
  G
  =
  \mathrm{diag}
  \left(
    c_{0}^{(0)},
    c_{0}^{(0)} + c_{1}^{(0)} p,
    \ldots,
    \sum_{l=0}^{m-2} c_{l}^{(0)} p^{l},
    \sum_{l=0}^{m-1} c_{l}^{(0)} p^{l}
    +
    \sum_{k=1}^{n}
    \left(
      \sum_{l=0}^{m-1} c_{l}^{(k)} p^{l}
    \right)
    x^{k}
  \right)
\end{equation}
where, by Theorem~\ref{centerEpm},
\[
  C_{k}
  =
  \mathrm{diag}
  \left(
    c_{0}^{(k)},
    c_{0}^{(k)} + c_{1}^{(k)} p,
    \ldots,
    \sum_{l=0}^{m-2} c_{l}^{(k)} p^{l},
    \sum_{l=0}^{m-1} c_{l}^{(k)} p^{l}
  \right).
\]

Thus, accordingly to expression (\ref{eqG}), for every possible choice of the element
\[
  C_{0}
  =
  \mathrm{diag}
  \left(
    c_{0}^{(0)},
    c_{0}^{(0)} + c_{1}^{(0)} p,
    \ldots,
    \sum_{l=0}^{m-2} c_{l}^{(0)} p^{l},
    \sum_{l=0}^{m-1} c_{l}^{(0)} p^{l}
  \right)
\]
in $Z(E_{p}^{(m)})$, there are as many elements $G$ as the number of elements of the set
\begin{equation} \label{eq0}
  \left\{
    \sum_{l=0}^{m-1} c_{l}^{(0)} p^{l}
    +
    \sum_{k=1}^{n}
    \left(
      \sum_{l=0}^{m-1} c_{l}^{(k)} p^{l}
    \right)
    x^{k}
    \ | \
    C_{1}, \ldots, C_{n} \in Z(E_{p}^{(m)})
  \right\}.
\end{equation}

We point out that in the preceding expression, $C_{1}, \ldots, C_{n}$ are not necessarily distinct and many of them could be the zero element in $E_p^{(m)}$.

But the set given in expression (\ref{eq0})
may be expressed as
\[
  \left\{
    \sum_{l=0}^{m-1} c_{l}^{(0)} p^{l}
    +
    \sum_{l=0}^{m-1}
    \left(
     \sum_{k=1}^{n} c_{l}^{(k)} x^{k}
    \right)
    p^{l}
    \ | \
    C_{1}, \ldots, C_{n} \in Z(E_{p}^{(m)})
  \right\}
\]
which gives all the elements of $\mathbb{Z}_{p^{m}}$ given that $C_{1}, \ldots ,C_{n}$ may be chosen in such a way that we get the $p$-adic decomposition of every element in $\mathbb{Z}_{p^{m}}$.
Therefore, for every element in $Z(E_{p}^{(m)})$, there are as many elements $G$ as elements in $\mathbb{Z}_{p^{m}}$, which results that we have $p^{m} \cdot p^{m} = p^{2m}$ different possibilities to select the element $G$.

An analogous result is obtained for $M = [a_{ij}]$ with
\[
  a_{ij}
  =
  \begin{cases}
    x,         & \text{for $i=j=m$}, \\
    y p^{m-1}, & \text{for $i=m$ and $j=1$}, \\
    0, & \text{otherwise}.
  \end{cases}
\]
and $y \in \mathbb{Z}_{p}$ is an element of order $p-1$.

On the other hand, in \cite{Maze2007} the authors show that the nearest the semigroup acting on the corresponding set is to be a group, i.e.\ the bigger the subgroup of units in the semigroup is, the easier is to develop a Polling-Hellman algorithm in the semigroup to compute the corresponding discrete logarithm.
In this case, as it was noted previously, the set of noninvertible elements in $E_{p}^{(m)}$ might be very close to the whole semigroup with an adequate election of $p$ and $m$.

An additional property of the preceding algorithm is the use of a different $G$ for every encryption.
In ElGamal cryptosystem \cite{ElGamal1985} the election of a random parameter in every encryption avoids an attack based on a pair plain text-encrypted text, given the existence of invertible elements.
The same idea cannot be developed in this case due to the almost nonexistence of invertible elements.
As in ElGamal case, the election of a different $G$ in every encryption avoids repetition attacks.

Now let $(R, T)$ be a public key, with $T = F \cdot R$, where $F \in \mathrm{Cen}(M)$, and let $A$ be a solution of the SAP, i.e., $A \cdot R = T$.
Suppose also that $G$ in step~\ref{sap.3} of Algorithm~\ref{sap} is such that $A G = G A$.
Then, given the encrypted message
\[
  (H, D)
  =
  (G \cdot R, S + G \cdot T)
\]
we have that
\[
  S + G \cdot T - (A G) \cdot R
  =
  S + (G A) \cdot R - (A G) \cdot R
  =
  S.
\]

Thus it is crucial that the element $G$ in step~\ref{sap.3} of Algorithm~\ref{sap} does not commute with the solution $A$ of the SAP.
Therefore to avoid that $G$ commute with any such solution, we must require that $G$ is not in the center of the multiplicative semigroup of $E_{p}^{(m)}$, which can be easily checked by requiring that there exists at least a nonzero element out of its main diagonal (see Theorem~\ref{centerEpm}).

Then suppose that an attacker tries to find $A \in Z(E_{p}^{(m)})$ being a solution of the SAP, i.e., $A \cdot R = T$.
Thus, again, according to Theorem~\ref{centerEpm}, $A$ must be of the form
\begin{equation} \label{eq3}
  A
  =
  \mathrm{diag}
  \left(
    a_{0},
    a_{0} + a_{1} p,
    \ldots,
    \sum_{j=0}^{m-2} a_{j} p^{j},
    \sum_{j=0}^{m-1} a_{j} p^{j}
  \right),
\end{equation}
with $a_{j} \in \mathbb{Z}_{p}$ for every $j = 0,1,2,\ldots,m-1$.

\begin{thm}
Assume that
\[
  R = (r_{0},r_{1},\ldots,r_{m-1})
  \quad \text{and} \quad
  T = (t_{0},t_{1},\ldots,t_{m-1})
\]
are elements in $\mathbb{Z}_{p} \times \mathbb{Z}_{p^{2}} \times \cdots \times \mathbb{Z}_{p^{m}}$ such that $A \cdot R = T$.
If $r_{j} \not\equiv_{p} 0$ for $j=0,1,2,\ldots,m-1$ then
\[
  r_{0}^{-1} t_{0}
  \equiv_{p}
  r_{1}^{-1} t_{1}
  \equiv_{p}
  \cdots
  \equiv_{p}
  r_{m-1}^{-1} t_{m-1}.
\]
\end{thm}
\begin{proof}
Since $A \cdot R = T$ it follows that
\begin{gather}
  a_{0} r_{0} \equiv_{p} t_{0},  \label{eq4} \\
  (a_{0} + p a_{1} + \cdots + p^{j} a_{j}) r_{j}
  \equiv_{p^{j+1}}
  t_{j},
  \quad
  \text{for $j=1,2,\ldots,m-1$}. \label{eq5}
\end{gather}

Now, from expression (\ref{eq4}) and the fact that $r_{0} \not\equiv_{p} 0$, we obtain that $a_{0} \equiv_{p} r_{0}^{-1} t_{0}$.

Assume now that $j=1,2,\ldots,m-1$.
From expression (\ref{eq5}) we have that
\[
  (a_{0} + p a_{1} + \cdots + p^{j} a_{j}) r_{j} - t_{j}
  =
  p^{j+1} h,
  \quad \text{for some $h \in \mathbb{Z}$},
\]
so, $p \mid (t_{j} - a_{0} r_{j})$.
Therefore, $a_{0} r_{j} \equiv_{p} t_{j}$ and using the fact that $r_{j} \not\equiv_{p} 0$, we obtain that $a_{0} \equiv_{p} r_{j}^{-1} t_{j}$.
\end{proof}

As an immediate consequence we get the following result.

\begin{cor} \label{conditionSAP}
Let $F \in E_{p}^{(m)}$ as in Algorithm~\ref{sap} defining a public key $(R,T)$, of some user, with $T = F \cdot R$, being $R = (r_{0},r_{1},\ldots,r_{m-1})$ and $T = (t_{0},t_{1},\ldots,t_{m-1})$ with $r_{j}, t_{j} \in \mathbb{Z}_{p^{j+1}}$, for $j = 0,1,2,\ldots,m-1$.
If $p \nmid r_{j}$ for every $j=0,1,2,\ldots,m-1$ and there exits $k$ such that $r_{k}^{-1} t_{k} \not \equiv_{p} r_{0}^{-1} t_{0}$, then an attacker cannot compute $A \in Z(E_{p}^{(m)})$ such that $T = A \cdot R$.
\end{cor}

\section{A public key cryptosystem based on the DP}
\label{PKCDP}

Our aim in this section is to give a trap-door function based on the DP in $E_{p}^{(m)}$ following ElGamal's ideas in the case of the DLP.
ElGamal \cite{ElGamal1985} introduced his cryptosystem based on the Diffie-Hellman key exchange
protocol \cite{Diffie1976}.
Thus we provide a key exchange in the ring $E_{p}^{(m)}$ following the CAKE concept introduced in
\cite[Definition~3]{Shpilrain2006b}.

To do so, given $M \in E_{p}^{(m)}$, we will consider the set
\begin{equation} \label{eq6}
  H(M)
  =
  \left\{
    \sum_{i=0}^{k} C_{i} M^{i}
    \ | \
    C_{i} \in Z(E_{p}^{(m)}), k \in \mathbb{Z}^{+}
  \right\}
\end{equation}
where $Z(E_{p}^{(m)})$ is the center of the semiring $E_{p}^{(m)}$ (see Theorem~\ref{centerEpm}) and $\mathbb{Z}^{+}$ denotes the set of positive integers.

\begin{prot}[DHDP protocol] \label{prot1}
Alice and Bob agree on two public elements $X,M \in E_{p}^{(m)}$ such that $M \notin \mathrm{Cen}(X)$.
\begin{enumerate}
\item Alice chooses two different elements $A_{1}, A_{2} \in H(M)$ and transmits $G_{A} = A_{1} X A_{2}$ to Bob.
\item Bob chooses $B_{1}, B_{2} \in \mathrm{Cen}(M)$ such that $B_{1} X \neq X B_{2}$ and transmits $G_{B} = B_{1} X B_{2}$ to Alice.
\item Alice computes $A_{1} G_{B} A_{2}$.
\item Bob computes $B_{1} G_{A} B_{2}$.
\end{enumerate}
\end{prot}

Now it is clear that both Alice and Bob share a common value and that they could exchange their roles in the way they choose their corresponding private information.
The latter is a particular case of \cite[Example 3]{Shpilrain2006b}, where the authors illustrate a protocol based on the DP in a general semigroup.

Trying to break the preceding protocol, in its general setting as it is shown in \cite[Example 3]{Shpilrain2006b} gives rise to the following problem directly related to the DP.

\begin{defn}
Let $G$ be a semigroup, $A, B \subseteq G$ two subsemigroups such that $a b = b a$ for every $a \in A$ and $b \in B$ and assume that $x \in G$.
The \textbf{DH Decomposition Problem (DHDP)} consists in given two elements $a_{1} x a_{2}$ and $b_{1} x b_{2}$, with $a_{1}, a_{2} \in A$ and $b_{1}, b_{2} \in B$ such that provide a DHDP key exchange as above, find the element $a_{1} b_{1} x b_{2} a_{2}$.
\end{defn}

It is immediate that being able to solve the DP implies that we will be able to solve the DHDP.

On the other hand, in \cite{Climent2011a} (see also \cite{Climent2011c,Climent2012c}) the authors introduce a key exchange protocol over the noncommutative ring $\mathrm{End}(\mathbb{Z}_{p} \times \mathbb{Z}_{p^{2}})$ that fits with the previous protocol.
In \cite{Kamal2012} the authors solve the DHDP problem using some invertible elements in $\mathrm{End}(\mathbb{Z}_{p} \times \mathbb{Z}_{p^{2}})$ resulting a cryptanalysis of that proposal.

From the key exchange based on the DP, and analogously to the ElGamal cryptosystem, we can derive a public key cryptosystem whose cryptanalysis is computationally equivalent to the DHDP.
We will give it in a general setting for semigroups, but let us introduce first the following notation.
Let $t \in \mathbb{Z}^{+}$ and consider a one-to-one map $\beta : G \longrightarrow \mathbb{Z}_{2}^{t}$; so, for $x \in G$, $\beta(x)$ is the binary representation of $t$ digits of $x$.
Moreover, we denote by $\oplus$ the bitwise \textsf{xor} operation in $\mathbb{Z}_{2}^{t}$.

\begin{prot}[ElGamal DP protocol (EGDP protocol)] \label{prot2}
Let $G$ be a semigroup and assume that $A, B \subseteq G$ are two subsemigroups such that $a b = b a$ for every $a \in A$ and $b \in B$.
Let $m$ be the message that Bob desires to send Alice.
\begin{enumerate}
\item \label{prot2.1}
Alice chooses $x \in G$ such that $x a \neq a x$, for every $a \in A$, and chooses $a_{1}, a_{2} \in A$.
\item \label{prot2.2}
Alice makes public the pair $(x, a_{1} x a_{2})$.
\item \label{prot2.3}
Bob chooses $b_{1}, b_{2} \in B$ randomly and such that $x b_{i} \neq b_{i} x$ for $i = 1,2$ and sends Alice the pair $(f,d) = (b_{1} x b_{2}, \beta^{-1}(\beta(m) \oplus \beta(b_{1} a_{1} x a_{2} b_{2}))$.
\item \label{prot2.4}
Alice recovers $m$ by computing $m = \beta^{-1}(\beta(d) \oplus \beta(a_{1} f a_{2}))$.
\end{enumerate}
\end{prot}

Part \ref{prot2.4} of the above protocol is correct because
\begin{align*}
  \beta(d) \oplus \beta(a_{1} f a_{2})
    & = \beta(d) \oplus \beta(a_{1} b_{1} x b_{2} a_{2}) \\
    & = \beta(\beta^{-1}(\beta(m) \oplus \beta(b_{1} a_{1} x a_{2} b_{2}))) \oplus \beta(a_{1} b_{1} x b_{2} a_{2}) \\
    & = \beta(m) \oplus \beta(b_{1} a_{1} x a_{2} b_{2}) \oplus \beta(a_{1} b_{1} x b_{2} a_{2}) \\
    & = \beta(m)
\end{align*}
since $\beta(b_{1} a_{1} x a_{2} b_{2}) = \beta(a_{1} b_{1} x b_{2} a_{2})$ as $b_{i} a_{i} = a_{i} b_{i}$ for $i = 1,2$.
Consequently,
\[
  \beta^{-1}(\beta(d) \oplus \beta(a_{1} f a_{2})) = \beta^{-1}(\beta(m)) = m.
\]

We point out that if we are considering a ring $R$ instead of a semigroup $G$, then we can substitute the map $\beta : G \longrightarrow \mathbb{Z}_{2}^{t}$ by the identity map in $R$ and the \textsf{xor} operation $\oplus$ by the addition in $R$.
Here we need the existence of a zero element in order to get $m = d - a_{1} f a_{2}$.

\begin{expl} \label{expl1}
\begin{enumerate}
\item \label{expl1.1}
If the semigroup $G$ is commutative, then we have that EGDP corresponds to the ElGamal public key cryptosystem defined by the action of $G$ over itself (cf.\ \cite{Maze2007}).
In this case, as we pointed out in Section~\ref{Intro}, the security of the cryptosystem is based on the difficulty of what the authors in \cite{Maze2007} called the SAP.

\item \label{expl1.2}
Let us consider the ring $E_{p}^{(m)}$.
Then Alice and Bob can agree in a public element $M \in E_{p}^{(m)}$.
Now let $S \in E_{p}^{(m)}$ be the secret that Bob wants to send Alice.
Then $S$ may be sent in a confidential manner in the following way:
\begin{enumerate}
\item Alice chooses $N \in E_{p}^{(m)}$ such that $N M \neq M N$ and two elements $A_{1}, A_{2} \in H(M)$ (see equation (\ref{eq6})), and publishes her public key $(N, A_{1} N A_{2})$.
\item Bob chooses randomly two other elements $B_{1}, B_{2} \in \mathrm{Cen}(M)$ (or simply $B_{1}, B_{2} \in H(M)$ as Alice) and sends her the pair given by $(F,D) = (B_{1} N B_{2}, S + B_{1} A_{1} N A_{2} B_{2})$.
\item Alice recovers $S$ by computing $D - A_{1} F A_{2}$ due to the fact that $A_{i}$ and $B_{i}$ commute for $i = 1,2$.
\end{enumerate}
\end{enumerate}
\end{expl}

The following result shows the relation between the cryptanalysis of the EGDP and DHDP protocols.

\begin{thm}
Breaking the EGDP protocol is equivalent to solve the DHDP.
\end{thm}
\begin{proof}
Assume that Eve can solve the DHDP and she wants to get $m$ from
\[
  (b_{1} x b_{2}, \beta^{-1}(\beta(m) \oplus \beta(b_{1} a_{1} x a_{2} b_{2}))).
\]
Since Eve is able to solve the DHDP, then she gets $b_{1} a_{1} x a_{2} b_{2}$ from $b_{1} x b_{2}$ and the Alice's public key $a_{1} x a_{2}$.
Thus she can recover
\[
  m
  =
  \beta^{-1}(\beta(m) \oplus \beta(b_{1} a_{1} x a_{2} b_{2}) \oplus \beta(b_{1} a_{1} x a_{2} b_{2})).
\]

Now assume that Eve can solve the EGDP.
Then she can obtain any message $m$ from the information $x, a_{1} x a_{2}, b_{1} x b_{2}, \beta^{-1} (\beta(m) \oplus \beta(b_{1} a_{1} x a_{2} b_{2}))$.
If Eve wishes to get $b_{1} a_{1} x a_{2} b_{2}$ from $x, a_{1} x a_{2}, b_{1} x b_{2}$, then she encrypts $m$ using $b_{1} x b_{2}$ as random element in step~\ref{prot2.3} of Protocol~\ref{prot2}, to get $d = \beta^{-1}(\beta(m) \oplus \beta(b_{1} a_{1} x a_{2} b_{2}))$ and thus, she gets the solution to the DHDP.
\end{proof}

Let us show how an analogous reasoning to the one used in the cryptanalysis introduced in \cite{Kamal2012} may apply to break the EGDP protocol.
Let us assume that Bob sends Alice the pair
\[
  (f,d)
  =
  (b_{1} x b_{2}, \beta^{-1}(\beta(m) \oplus \beta(b_{1} a_{1} x a_{2} b_{2})))
\]
as in the EGDP algorithm and that Eve is able to find $w_{1}, w_{2}$ commuting with every element in $A$ and such that $f w_{2} = w_{1} x$.
Suppose also that $w_{2}$ is invertible in the semigroup $G$.
Then
\[
  w_{1} a_{1} x a_{2} w_{2}^{-1}
  =
  a_{1} w_{1} x w_{2}^{-1} a_{2}
  =
  a_{1} f a_{2}
\]
and thus $\beta^{-1}(\beta(d) \oplus \beta(w_{1} a_{1} x a_{2} w_{2}^{-1})) = m$.

Therefore, the further the subgroup of units in $G$ is from being $G$, the more difficult will be to apply this reasoning to break the EGDP protocol.
This applies to the case of $E_{p}^{(m)}$ since as we mentioned previously, a suitable choice of $p$ and $m$ provides a ring such that almost all its elements are not units.

We will finish this paper by giving some conditions on the public key used for the EGDP protocol in the case of the ring $E_{p}^{(m)}$ in order to avoid an attack by reducing it to find a solution to a certain SAP.

Let us assume that Alice's public key is given by the pair $(X, A_{1} X A_{2})$ and let Eve be an attacker.
If Eve is able to find $M \in Z(E_{p}^{(m)})$ such that it is a solution of the SAP $M X = A_{1} X A_{2}$, then Eve may choose $H \in Z(E_{p}^{(m)})$ and invertible, say
\[
  H
  =
  \mathrm{diag}
  \left(
    h_{0},
    h_{0} + h_{1} p,
    \ldots,
    \sum_{l=0}^{m-2} h_{l} p^{l},
    \sum_{l=0}^{m-1} h_{l} p^{l}
  \right),
\]
with $h_{0} \not \equiv_{p} 0$ (cf.\ \cite{Climent2014d}).
Thus Eve writes $M H^{-1} H$ and gets that
\[
  A_{1} X A_{2}
  =
  M X
  =
  M H^{-1} X H,
\]
solving the DP and therefore she solves the DHDP.

As previously noted, the next result gives conditions on the public key of the EGDP protocol that avoid breaking it by reducing the DP to a SAP as above explained.
Its proof is a direct application of Corollary~\ref{conditionSAP}.

\begin{cor}
In the precedent situation, let $(X,P)$, with $P = A_{1} X A_{2}$, be the public key of some user of the EGDP protocol over the ring $E_{p}^{(m)}$.
If there exists $k = 1,2,\ldots,m$ such that $X_{i,k} \not\equiv_{p} 0$ for every $i = 1,2,\ldots,m$ and there exists $j$ such that $X_{j,k}^{-1} P_{j,k} \not \equiv_{p} X_{1,k}^{-1} P_{1,k}$, then an attacker cannot break the EGDP by reducing it to solve a SAP over the ring $E_{p}^{(m)}$.
\end{cor}

\section*{Acknowledgments}

The authors would like to thank the reviewers their valuable comments as well as some references that helped to enhance the paper.


%

\end{document}